\documentclass[twocolumn]{IEEEtran}

\usepackage{amsmath,epsfig,amssymb,verbatim,amsopn,cite}
\usepackage{amsthm}
\usepackage{balance}
\usepackage{multirow}
\usepackage[usenames,dvipsnames]{color}
\usepackage[all]{xy}  
\usepackage{url}
\usepackage{amsfonts}
\usepackage{amssymb}
\usepackage{epsfig}
\usepackage{epstopdf}
\usepackage{bm}
\usepackage{balance}
\usepackage{graphicx}
\usepackage{footnote}
\usepackage{cancel}
\usepackage{algorithm,algorithmic}

\usepackage{makecell} 
\usepackage[short]{optidef} 

\newtheorem{Lemma}{Lemma}

  {\proof}{\proofend}
\newtheorem{proposition}{Proposition}

\usepackage[margin=15.5mm,top=17.8mm,bottom=26.2mm]{geometry}

\newcommand{\qg}{{ \textbf{g} }}

\newcommand{\qn}{{\bf n}}

\newcommand{\qw}{{\bf w}}
\newcommand{\qx}{{\bf x}}

\newcommand{\qz}{{ \textbf{z} }}

\newcommand{\qA}{{\bf A}}

\newcommand{\qF}{{\bf F}}

\newcommand{\qH}{{ \textbf{H} }}
\newcommand{\qI}{{ \textbf{I} }}

\newcommand{\qN}{{\bf N}}

\newcommand{\qY}{{\bf Y}}

\newcommand{\ettall}{\emph{et al.}}

\newcommand{\UE}{\mathtt{I}}
\newcommand{\sn}{\mathtt{E}}

\DeclareMathOperator{\ETAI}{\boldsymbol{\eta}^{\mathtt{I}}}
\DeclareMathOperator{\ETAE}{\boldsymbol{\eta}^{\mathtt{E}}}

\DeclareMathOperator{\K}{\mathcal{K}}

\DeclareMathOperator{\M}{\mathcal{M}}

\DeclareMathOperator{\C}{\mathbb{C}}

\DeclareMathOperator{\CN}{\mathcal{CN}}

\newcommand{\MRT}{\mathsf{MRT}}

\newcommand{\wimx}{\qw_{\mathtt{i},mk}}
\newcommand{\wimk}{\qw_{\mathtt{I},mk_{i}}}

\newcommand{\wemj}{\qw_{\mathtt{E},m k_{e}}}
\newcommand{\wemjp}{\qw_{\mathtt{E},m k'_{e}}}

\newcommand{\wimkp}{\qw_{\mathtt{I},mk_{i}'}}

\newcommand{\Snn}{\sigma_n^2}
\newcommand{\Ex}{\mathbb{E}}

\newcommand{\yej}{y_{\mathtt{E}, k_e}}

\newcommand{\yik}{y_{\mathtt{I},k_i}}

\newcommand{\gmkiu}{\qg_{mk_{i}}^{\UE}}

\newcommand{\gmjue}{\qg_{mk_{e}}^{\sn}}

\newcommand{\FmRIS}{\qF_{m}}

\newcommand{\gejue}{\qg_{m k_{e}}^{\sn}}

\newcommand{\trace}{\mathrm{tr}}
\newcommand{\diag}{\mathrm{diag}}

\newcommand{\snrul}{\rho_{u}}
\newcommand{\snrdl}{\rho_{d}}

\newcommand{\zmk}{\qz_{mk}}
\newcommand{\zmklos}{\bar{\qz}_{mk}}
\newcommand{\zmkilos}{\bar{\qz}_{m k_i}}

\newcommand{\zmknlos}{\tilde{\qz}_{mk}}
\newcommand{\zmkpnlos}{\tilde{\qz}_{mk'}}

\newcommand{\gmk}{\qg_{mk}^{\mathtt{i}}}

\newcommand{\gmkbar}{\bar{\qg}_{mk}^{\mathtt{i}}}

\newcommand{\gmkpbar}{\bar{\qg}_{mk'}^{\mathtt{i}}}

\newcommand{\hgmk}{\hat{\qg}_{mk}^{\mathtt{i}}}
\newcommand{\hgmki}{\hat{\qg}_{m k_i}^{\mathtt{I}}}

\newcommand{\hgmkp}{\hat{\qg}_{mk'}^{\mathtt{i}}}

\DeclareMathOperator{\OO}{\mathcal{O}}

\newcommand{\tilgmk}{\tilde{\qg}_{mk}^{\mathtt{i}}}

\newcommand{\betamk}{\beta_{mk}^{\mathtt{i}}}
\newcommand{\barbetamk}{\bar{\beta}^{\mathtt{i}}_{mk}}
\newcommand{\barbetamkp}{\bar{\beta}^{\mathtt{i}}_{mk'}}
\newcommand{\barbetamki}{\bar{\beta}^{\mathtt{I}}_{m k_i}}
\newcommand{\barbetamke}{\bar{\beta}^{\mathtt{E}}_{m k_e}}

\newcommand{\DSki}{\mathsf{DS}_{k_i}}
\newcommand{\BUki}{\mathsf{BU}_{k_i}}
\newcommand{\IUIki}{\mathsf{IUI}_{k_{i} {k_{i}}'}}
\newcommand{\EUIki}{\mathsf{EUI}_{k_{i} k_{e}}}
\newcommand{\UIki}{\mathsf{UI}_{k_{i} k'}}

\newcommand{\Covhatgmk}{\boldsymbol{\Sigma}_{\hgmk}}

\newcommand{\SINRki}{\mathrm{SINR}_{k_i}}

\newcommand{\etamkI}{\eta_{m k_{i}}^{\mathtt{I}}}

\newcommand{\etamkpI}{\eta_{m k_{i}'}^{\mathtt{I}}}
\newcommand{\etamjE}{\eta_{m k_{e}}^{\mathtt{E}}}

\newcommand{\etamjpE}{\eta_{m k_{e}'}^{\mathtt{E}}}

\newcommand{\xik}{x_{\mathtt{I},k_{i}}}
\newcommand{\xikp}{x_{\mathtt{I},k'_i}}
\newcommand{\xej}{x_{\mathtt{E},k_{e}}}
\newcommand{\xejp}{x_{\mathtt{E},k'_e}}

\DeclareMathOperator{\PHI}{\boldsymbol{\Phi}}

\DeclareMathOperator{\VARPHI}{\boldsymbol{\varphi}}

\begin{document}

\title{SWIPT in Cell-Free Massive MIMO Using Stacked Intelligent Metasurfaces}

\author{Thien Duc Hua, Mohammadali Mohammadi, Hien Quoc Ngo, and  Michail Matthaiou\\
\small{
Centre for Wireless Innovation (CWI), Queen's University Belfast, U.K.\\
Email:\{dhua01, m.mohammadi, hien.ngo, m.matthaiou\}@qub.ac.uk 
}}\normalsize
\allowdisplaybreaks



\maketitle

\begin{abstract}
We investigate the integration of stacked intelligent metasurfaces (SIMs) into cell-free massive multiple input multiple output  (CF-mMIMO) system to enhance the simultaneous wireless information and power transfer (SWIPT) performance. Closed-form expressions for the spectral efficiency (SE) of the information-decoding receivers (IRs) and the average sum of harvested energy (sum-HE) at the energy-harvesting receivers (ERs) in the novel system model are derived to subsequently formulate a maximum total average sum-HE problem under a minimum SE threshold per each IR. This problem jointly optimizes the SIM phase-shift (PS) configuration and access points' (APs) power allocation, relying on long-term statistical channel state information (CSI). This non-convex problem is then transformed into more tractable forms. Then, efficient algorithms are proposed, including a layer-by-layer heuristic method for SIMs PS configuration that prioritizes sum-HE for the ERs and a successive convex approximation (SCA)-based power allocation scheme to improve the achievable SE for the IRs. Numerical results show that our proposed algorithms achieve an almost 7-fold sum-HE gain as we increase the number of SIM layers, while the proposed power allocation (PPA) scheme often gains up to 40\% in terms of the achievable minimum SE, compared to the equal power allocation.
\let\thefootnote\relax\footnotetext{This work was supported by the U.K. Engineering and Physical Sciences Research Council (EPSRC) (grants No. EP/X04047X/1 and EP/X040569/1). The work of T. D. Hua and M. Matthaiou was supported by the European Research Council (ERC) under the European Union’s Horizon 2020 research and innovation programme (grant agreement No. 101001331). The work of H.~Q.~Ngo was supported by the U.K. Research and Innovation Future Leaders Fellowships under Grant MR/X010635/1, and a research grant from the Department for the Economy Northern Ireland under the US-Ireland R\&D Partnership Programme.}

\end{abstract}


\section{Introduction}
The growing demand for instantaneous and reliable wireless connectivity, combined with the need for ubiquitous energy harvesting in Internet-of-Everything (IoE) devices, has positioned SWIPT as a key area of interest in both academic and industrial research. Over the past decade, extensive studies have highlighted the effectiveness of integrating SWIPT with CF-mMIMO systems, demonstrating notable improvements in both HE and SE metrics \cite{2024:Mohammadi:survey}. 
However, a significant limitation of current CF-mMIMO SWIPT systems lies in the reliance on conventional transmit antennas at the APs. On one hand, an insufficient number of antennas can increase the probability of losing line-of-sight (LoS) propagation, which would adversely affect the SWIPT performance. On the other hand, deploying an excessive number of antennas at APs leads to prohibitive power consumption and increased inter-user interference due to high spatial correlation and pilot contamination (PC)~\cite{2025:Hien:invited}. To mitigate interference, fully digital beamforming techniques, like zero-forcing (ZF), are commonly used. However, these methods come with high hardware cost and computational complexity \cite{2024:an:nearfieldsim}. As a result, there is an urgent need for a CF-mMIMO SWIPT system that leverages innovative technologies to enable efficient and power-saving operation.

Reconfigurable intelligent surfaces (RISs) are promising architectures for enhancing the wireless networks' performance. Specifically, SIMs, consisting of multiple layers of reconfigurable surface arrays, have gained significant attention compared to conventional single-layer RISs due to their superior signal processing capabilities. Leveraging SIMs in the vicinity of APs (with a small number of antennas and radio frequency (RF) chains) enables the formation of strong LoS connections while ensuring acceptable energy consumption~\cite{haoliu:arxiv:2024:SIMHardwareSurvey}. Additionally, SIMs facilitate the development of practical, low-cost prototypes, offering a scalable and efficient solution~\cite{An:Arxiv:2023}. Previous studies have investigated the integration of SIMs in various scenarios. An~\ettall~\cite{An:SIMHMIMO:JSAC:2023} designed a pair of SIM-attached transceivers to achieve interference-free transmission. Perovic~\ettall~\cite{Perovic:CL:2024} investigated the channel cutoff rate and the mutual information problem under MIMO systems with SIM-attached APs. Li~\ettall~\cite{Li:TC:2024} and Shi~\ettall~\cite{shi:arxiv:2024} respectively developed a gradient descent algorithm and a cyclic search with bisection scheme to maximize the achievable SE with SIM-attached APs in a CF-mMIMO system. To our knowledge, none of the studies mentioned above have investigated SWIPT CFmMIMO systems with novel SIM-attached APs design and long-term statistical CSI, PC, and pratical non-linear energy harvesting (NL-EH) model. To address this gap, we propose a SIMs-assisted SWIPT CFmMIMO system including:

\begin{itemize}
    \item We introduce a SIM-assisted CF-mMIMO SWIPT system for a Ricean channel model and provide the second- and fourth-order moments of the actual channels and estimated channels under the PC effect.
    \item Closed-from expressions for the achievable SE of the IRs and the average sum HE of the ERs with a NL-EH model are derived.
    \item We propose a layer-by-layer heuristic PS search and power allocation scheme to maximize the total average sum-HE, while satisfying a minimum SE per IR. Our numerical results demonstrate that the ER benefits significantly from increasing the SIM layers, while the achievable minimum SE increases about $40\%$ by applying the PPA scheme. 
\end{itemize}

\textit{Notation:} 
The superscript $(\cdot)^\dag$ stands for the Hermitian-transpose; $\mathbf{I}_N$ denotes the identity matrix $N\times N$;  $\text{mod}( \cdot, \cdot)$ is the modulus operation; $ \lceil \cdot \rceil$ denotes the ceil function;  $\Vert\cdot\Vert^{2}_{\mathrm{F}}$ denotes the Frobenius norm square of vector/matrix; $\trace(\cdot)$ denotes the trace operator; 
a circular symmetric complex Gaussian variable with variance $\sigma^2$ is denoted by $\mathcal{CN}(0,\sigma^2)$. Finally, $\mathbb{E}\{\cdot\}$ denotes the statistical expectation.
\vspace{-0.7em}
\section{System Model}~\label{sec:Sysmodel}
Consider a SIM-assisted CF-mMIMO SWIPT system, where all APs simultaneously serve $K_{i}$ IRs and $K_{e}$ ERs over the same frequency band. Each AP $m$ is equipped with $N$ antennas while all IRs and ERs deploy a single antenna. For notational simplicity, we define the sets $\K_{\mathtt{I}} \triangleq \{1,\dots,K_{i}\}$, $\K_{\mathtt{E}} \triangleq\{1,\ldots,K_{e}\}$, and $\M \triangleq\{1,\ldots,M\}$  to collect the indices of the IRs, ERs, and APs, respectively, and $\K \triangleq \{ \K_{\mathtt{I}}, \K_{\mathtt{E}} \}$. Each AP $m$ is implemented with a SIM in a way that the centers of the antenna arrays align with those of all metasurfaces of the SIM. 
Each SIM is a closed vacuum container comprising $L$ metasurfaces $\big(\mathcal{L} \triangleq \{1, \dots, L  \} \big)$, where each metasurface comprises $S$  PS elements $\big( \mathcal{S} \triangleq \{1, \dots, S \} \big)$~\cite{An:SIMHMIMO:JSAC:2023}.
We denote the attached $\mathcal{L}$-layer SIM at AP-$m$ as $\boldsymbol{\Psi}^{m} \triangleq \{ \boldsymbol{\Phi}^{m \ell}, \forall \ell \in \mathcal{L} \}$, where $\boldsymbol{\Phi}^{m \ell} \in \C^{S \times S} \triangleq \diag( e^{j\theta^{m \ell}_{1}}, \ldots, e^{j\theta^{m \ell}_{s}}, \ldots, e^{j\theta^{m \ell}_{S}} )$, with the PS element $\theta^{m \ell}_{s} $ within $ [0, 2 \pi]$.

\vspace{-0.8em}
\subsection{Channel Model}
We assume a block fading channel that remains invariant and frequency-flat during each coherence interval $\tau_{c}$, including $\tau$ samples for UL channel estimation training and the rest for downlink (DL) transmission. For a general interpretation, let $k \in \K$ be the index of the receivers and $\mathtt{i}\in\{\mathtt{I}, \mathtt{E}\}$ be the index corresponding to wireless information transmission (WIT) and wireless power transmission (WPT) operations. 
In each coherence interval, the instantaneous channel from AP $m$ to receiver $k$ is mathematically expressed as $\gmk \triangleq \FmRIS^{\dag} \zmk$, where $\FmRIS \in \C^{S \times N}$ and $\zmk \in \C^{S \times 1}$ are the SIM-effected aggregate channel of AP $m$, and channel from the last SIM layer of AP $m$ to receiver $k$, respectively.
In particular, the SIM-effected aggregated channel is formulated as
\begin{equation}~\label{eq:Fm}
\FmRIS \triangleq 
        \boldsymbol{\Phi}^{m L} \qH^{m L} 
        \cdots 
        \boldsymbol{\Phi}^{m \ell} \qH^{m \ell}
        \cdots 
        \boldsymbol{\PHI}^{m 1} \qH^{m 1}, 
\end{equation}
where $\qH^{m 1} \in \C^{S \times N}$ and $\qH^{m \ell} \in \C^{S \times S}$ are the EM-wave propagation matrices from AP $m$ to the first SIM layer and from the $(\ell - 1)$-th layer to the $\ell$-th layer, $\forall \ell > 1$, respectively. According to Rayleigh-Sommerfeld diffraction theory~\cite{An:SIMHMIMO:JSAC:2023}, the coefficient from the $\breve{s}$-th element of the $(\ell - 1)$-th layer to the  $s$-th element of the $(\ell - 1)$-th layer is formulated as
\begin{equation}~\label{eq:rayleigh_sommerfeld_coeffi}
[\qH^{m \ell}]_{s,\breve{s}} \!=\! \frac{\lambda^2 \cos(\chi^{m \ell}_{s,\breve{s}})}{4 d^{\ell}_{s,\breve{s}}}\bigg(\frac{1}{2\pi d^{\ell}_{s,\breve{s}}} \!-\! \frac{j}{\lambda} \bigg) \exp\Bigg( \frac{j 2 \pi d^{\ell}_{s,\breve{s}}}{\lambda} \Bigg),
\end{equation}
where $\lambda$ is the wavelength, $\chi^{m \ell}_{s,\breve{s}}$ is the angle between the propagation and normal directions of the $(\ell - 1)$-th layer, while $d^{\ell}_{s,\breve{s}}$ is the geometric distance from the $\breve{s}$-th element of the $(\ell - 1)$-th layer to the $s$-th element of the $\ell$-th layer, i.e.,
\begin{equation}~\label{eq:3D_distance}
d^{\ell}_{s,\breve{s}} = \sqrt{\Big( d_{\mathtt{PS}} \sqrt{(s_{z} - \breve{s}_{z} )^2 + (s_{y} - \breve{s}_{y})^2} \Big)^{2} + d_{\mathtt{SIM}}^{2} },
\end{equation}
where $d_{\mathtt{PS}}$ and $d_{\mathtt{SIM}}$ are the spacings between two adjacent elements and metasurfaces, respectively; the indices of the elements $s_{z}$ and $s_{y}$ along the $z$- and $y$-axes can be computed as $s_{z} \triangleq \lceil s/ \sqrt{S} \rceil$ and $s_{y} = \mod(s-1, \sqrt{S}) + 1$, respectively.

We note that the propagation coefficient from the $n$-th transmit antenna to the $s$-th element of the first SIM layer $[\qH^{m, 1}]_{s,n}$ is formulated by~\eqref{eq:rayleigh_sommerfeld_coeffi}. Subsequently, we adopt the Ricean channel model to characterize the instantaneous channel from the last SIM layer of AP $m$ to receiver $k$ as $ \zmk \triangleq \barbetamk\big(\sqrt{\kappa}\zmklos + \zmknlos \big)$,
where $\barbetamk \triangleq \betamk/(1+\kappa)$, with $\betamk$ representing the large-scale fading coefficient,  $\kappa$ is the Ricean factor, $\zmknlos$ is the NLoS component, whose $s$-th element  is distributed as $[\zmknlos]_{s} \sim \CN(0,1), \forall m \in \M \text{, } \forall k \in \K$. In addition,  $\zmklos$  is the LoS component whose $s$-th element  is given by
\begin{align}
    \big[ \zmklos \big]_{\!s} 
    &\!=\! \exp \!\Big( 
    \!
    \zeta s_{z} \sin{\big(\chi^{L}_{s,k} \big)} 
    \!+\!
    \zeta s_{y} \sin\big(\varepsilon^{L}_{s,k}\big) \cos\big(\chi^{L}_{s,k}\big) \Big),
\end{align}
where $\zeta \triangleq j 2 \pi d_{\mathtt{PS}} / \lambda$, $\sin\big(\varepsilon^{L}_{s,k}\big) \cos\big(\chi^{L}_{s,k}\big) = (n_{z} \!-\! k_{z})/d^{L}_{n,k}$, $\sin{\big(\chi^{L}_{s,k} \big)} = (\vert n_{z} - k_{z} \vert)/\big(d^{L}_{n,k} \big)$, 
where $d^{L}_{n,k}$ is the geometric distance from receiver $k$ to  element $n$ of the last SIM layer, while $\chi^{L}_{n,k}$ and $\varepsilon^{L}_{n,k}$ are the elevation and azimuth angle between the propagation direction and the $Oxy$ and $Oyz$ surfaces, respectively. To facilitate subsequent derivations, we present the channel statistics as $\Ex\{\Vert \gmk \Vert^{2} \} \!=\! \kappa \barbetamk \trace\big( \FmRIS \FmRIS^{\dag} \zmklos^{\dag} \zmklos  \big)
\!\!+\!\! \barbetamk \trace\big( \FmRIS \FmRIS^{\dag} \big)$, where $\gmk \sim \CN \Big(\gmkbar, \barbetamk \FmRIS^\dag \FmRIS \Big)$ with $\gmkbar \triangleq \sqrt{\barbetamk \kappa} \FmRIS^{\dag} \zmklos$.


\vspace{-1.0em}
\subsection{Uplink Training and Channel Estimation}\label{phase:ULforCE}
In the training phase, all IRs and ERs transmit their pilot sequences of length $\tau$ symbols to the APs.
We consider the general case where shared pilot sequences are used and denote $\mathcal{P}_k \subset \K $ as the set of receivers $k'$, including $k$, that are assigned with the same pilot sequence as the receiver $k$.
In this regard, $i_{k} \in \{1, \ldots, \tau \}$ denotes the index of the pilot sequence used by receiver $k$. The pilot sequences $\VARPHI_{i_k} \in \C^{\tau \times 1}$, with $\Vert \VARPHI_{i_k} \Vert^{2} = 1$, are mutually orthogonal so that $\VARPHI_{i_{k'}}^{H}\VARPHI_{i_k}=1 \text{ if } i_{k'} \in \mathcal{P}_k$ and $\VARPHI_{i_{k'}}^{H}\VARPHI_{i_k}=0,\text{ otherwise}$.
Then, the AP $m$ receives
\begin{equation}~\label{eq:receivedpilotsequence}
    \qY_{p,m}^{\mathtt{i}} \!=\! \sqrt{\tau\snrul }\big( \gmk  \VARPHI_{i_k}^\dag  \!+\!\sum\nolimits^{K}_{k'=1 \setminus k}\!\qg_{mk'}  \VARPHI_{i_{k'}}^\dag\big) \!+\! \qN_{p,m},
\end{equation}
where $\snrul$ is the normalized UL signal-to-noise ratio (SNR), while $\qN_{p,m} \in \C^{N \times \tau}$ is the receiver noise matrix containing independent
and identically distributed (i.i.d.) $\CN(0,\Snn)$ random variables (RVs). By invoking~\cite{Hua:WCNC:2024}, $\gmk$ is estimated using the linear minimum mean-squared error (MMSE) technique as
\vspace{-0.4em}
\begin{align}
    \hgmk 
    &\!= \!\sqrt{\barbetamk \kappa} \FmRIS^{\dag} \zmklos + 
    \qA_{mk}
    \nonumber\\
    & \hspace{2em} \times
    \Big(\!  \sqrt{\tau \snrul}
        \sum\nolimits_{k' \in \mathcal{P}_k } 
        \!\! { \sqrt{ \barbetamkp} \FmRIS^{\dag} \zmkpnlos + \Tilde{\qn}_{p,mk} } 
        \!\Big),
\end{align}
where $\Tilde{\qn}_{p,mk} \triangleq \qN_{p,m} \VARPHI_{i_k} \sim \CN(\boldsymbol{0},\Snn \qI_{N})$, and
$$\qA_{mk}\!\triangleq\!\! \big(\sqrt{\tau \snrul} \!\barbetamk \FmRIS^{\dag} \FmRIS \big)  \big( \tau \snrul \!\sum\nolimits_{k' \in \mathcal{P}_k}{\!\! \barbetamkp }  \FmRIS^{\dag} \FmRIS \!+\! \Snn \qI_{N}\big)^{-1}.$$
Thus, the statistical distribution of $\gmk$ is the following $\hgmk \sim \CN \big( \gmkbar, \Covhatgmk \big)$, where $\Covhatgmk \triangleq \sqrt{\tau \snrul} \barbetamk \FmRIS^{\dag} \FmRIS \qA_{mk}$. In addition, we obtain
\begin{equation}
    \Ex\{\Vert \hgmk \Vert^2 \} \triangleq  \kappa \barbetamk \trace(\FmRIS \FmRIS^\dag \zmklos \zmklos^\dag ) + \gamma^{i}_{mk},
\end{equation}
where
\vspace{-0.5em}
\begin{equation}
    \gamma^{i}_{mk} = \frac{\tau \snrul (\barbetamk)^2 \big(\trace(\FmRIS \FmRIS^{\dag})\big)^2 }
    {\tau \snrul \sum_{k' \in \mathcal{P}_k}{ \barbetamkp } \trace(\FmRIS \FmRIS^{\dag}) + N \Snn }.
\end{equation}
The expectation of the norm square of the estimate error $\tilgmk \triangleq \gmk - \hgmk$ is $\Ex \{ \Vert \tilgmk \Vert^2 \}
    = \barbetamk \trace(\FmRIS \FmRIS^{\dag}) - \gamma^{i}_{mk}$.

\vspace{-1.8em}
\subsection{Downlink Wireless Information and Power Transmission}
\vspace{-0.4em}
Define $\mathcal{M}{\mathtt{I}}\subset \mathcal{M}$ and $\mathcal{M}{\mathtt{E}}\subset \mathcal{M}$ as the sets of APs providing DL service for the IRs and ERs, respectively, with $\mathcal{M} = \mathcal{M}{\mathtt{I}} \cup \mathcal{M}{\mathtt{E}}$. Let $\xik$ ($\xej$) denote the information (energy) symbol transmitted to  IR $k_i$ (ER $k_e$)  that satisfies $\Ex\{ \vert \xik \vert^2 \} = \Ex\{ \vert \xej \vert^2 \} =1$. The IR and ER receive
\vspace{-0.1em}
\begin{subequations}
 \begin{align}
    \yik 
    &\!=\!  
    n_{d,k_{i}} 
    \!\!+\!\!
    \sum\nolimits_{k'_i\in \K_{\mathtt{I}} } \sum\nolimits_{m \in \M_{\mathtt{I}} } \sqrt{\rho_d\etamkpI} (\gmkiu)^\dag \wimkp \xikp \nonumber \\
    &\hspace{-2em}+ \! \sum\nolimits_{k_e\in\K_{\mathtt{E}}}  \! \sum\nolimits_{m \in \M_{\mathtt{E}} } \!\! \sqrt{ \rho_d\etamjE} (\gmkiu)^\dag\! \wemj \xej,\\
    \yej
    &\!=\!  
    n_{d,k_{e}} 
    \!\!+\!\!
    \sum\nolimits_{k'_e\in \K_{\mathtt{E}} }\!  \sum\nolimits_{m \in \M_{\mathtt{E}}}\!\!
    \sqrt{ \rho_d\etamjpE} (\gmjue)^\dag \wemjp \xejp    \nonumber\\
    &\hspace{-2em}+\!  
    \sum\nolimits_{k_i \in\K_{\mathtt{I}} }\!
    \sum\nolimits_{m\in\M_{\mathtt{I}} }\!\! 
    \sqrt{ \rho_d\etamkI} (\gmjue)^\dag \wimk \xik \!,
\end{align}
\end{subequations}
respectively, where $n_{d,k_{i}},n_{d,k_{e}} \!\sim \! \CN(0,1)$ are the additive noise terms, $\snrdl = \tilde{\rho}_{d}/\Snn$ is the normalized DL SNR; $\wimk \in \C^{N\times 1}$ ($\wemj\in \C^{N\times 1}$) represents the precoding vector for IR $k_i$ (ER $k_e$) with $\Ex\big\{\big\Vert\wimk\big\Vert^2\big\}=\Ex\big\{\big\Vert\wemj\big\Vert^2\big\}=1$. 

To this end, we apply the normalized maximum ratio transmission (MRT) precoder since power transfer with MRT is shown to be optimal for WPT systems, especially with a large number of antennas~\cite{cite:MRT_for_HE:Almradi}. Mathematically speaking: $\wimx^{\MRT} =  \Big(\!\sqrt{\alpha^{\MRT}_{mk}}\Big)^{-1}\hgmk$, where
$\alpha^{\MRT}_{mk} \triangleq \Ex\{\Vert \hgmk \Vert^2 \} =  \kappa \barbetamk \trace(\FmRIS \FmRIS^\dag \zmklos \zmklos^\dag ) + \gamma^{i}_{mk}.$

Applying the use-and-then-forget technique~\cite{cite:HienNgo:cf01:2017}, the signal received at the $k_i$-th IR can be expressed as
\vspace{-0.1em}
\begin{align}~\label{eq:yi:hardening}
    \yik &=  \DSki  \xik +
    \BUki \xik 
         +\sum\nolimits_{k_{i}'\in \K_{\mathtt{I}} \setminus k_{i}}
         \!\!
     \IUIki
     \xikp\nonumber\\
    &\hspace{1em}
    + \sum\nolimits_{k_{e}\in\K_{\mathtt{E}} }
     \EUIki \xej + n_{k_i},~\forall k_i\in\K_{\mathtt{I}},
\end{align}
where $\DSki$, $\BUki$, $\IUIki$, and $\EUIki$ represent the desired signal, the beamforming gain uncertainty, the interference cause by the $k'_i$-th IR, and the interference caused by the $k_e$-th ER, respectively, given by
\begin{subequations}
  \begin{align}~\label{eq:yi:components}
    \DSki &\triangleq \sum\nolimits_{m\in\M_{\mathtt{I}}} \sqrt{ \snrdl \etamkI} \Ex \big\{(\gmkiu)^\dag \wimk^{\MRT} \big\},  
    \\
    \BUki &\triangleq \sum\nolimits_{m\in\M_{\mathtt{I}}} \sqrt{\snrdl \etamkI} \Bigl( (\gmkiu)^\dag \wimk^{\MRT} \nonumber \\
    &\hspace{1em}- \Ex \big\{(\gmkiu)^\dag \wimk^{\MRT} \big\} \Bigl)~\label{eq:component_BUk}, 
    \\
    \IUIki &\triangleq \sum\nolimits_{m\in\M_{\mathtt{I}}} \sqrt{\snrdl \etamkpI} (\gmkiu)^\dag \wimkp^{\MRT},~\label{eq:IU_interference}
    \\
    \EUIki &\triangleq \sum\nolimits_{m\in\M_{\mathtt{E}}} \sqrt{ \snrdl \etamjE} (\gmkiu)^\dag \wemj^{\MRT}.  
    \end{align}  
\end{subequations}

Thus, the DL SE in (bit/s/Hz) for IR $k_i$ is
\begin{align}~\label{eq:SEk:Ex}
    \mathrm{SE}_k
      &=
      \Big(1\!- \!\frac{\tau}{\tau_c}\Big)
      \log_2
      \Big(
       1\! + \SINRki^{\MRT}\big(\boldsymbol{\Theta}, \ETAI, \ETAE\big)
     \Big),
\end{align}
where $\ETAI = [\eta_{m1}^{\mathtt{I}}, \ldots, \eta_{m K_i}^{\mathtt{I}}]$, $ \ETAE = [\eta_{e 1}^{\mathtt{E}}, \ldots, \eta_{e K_e}^{\mathtt{E}}]$, and $\SINRki$ is the effective SINR at IR $k_i$,  given by
\vspace{0.2em}
\begin{align}~\label{eq:SINE:general}
       \SINRki^{\MRT}= \frac{
                 \vert  \DSki  \vert^2
                 }
                 {  
                 \Ex\{\vert  \BUki  \vert^2\} 
                 +\mathsf{UI}
                   \!+\!  1},
\end{align}
where $\mathsf{UI}\!\triangleq \!\!\sum_{k_{i}'\in\K_{\mathtt{I}} \setminus k_{i}}\! \!\Ex\{\! \vert \IUIki \vert^2\!\} \! + \! \sum_{k_e\in \K_{\mathtt{E}}} \!\Ex \{\vert  \EUIki \vert^2\}$.
                   
To characterize the HE at ER $k_e$, a NL-EH model with the sigmoidal function is used~\cite{Hua:WCNC:2024}. Subsequently, the total HE at ER $k_e \in \K_{\mathtt{E}}$ is given by $\mathcal{E}^{\mathrm{NL}}_{k_e} = \big( \Lambda\big(\mathcal{E}_{k_e}\big) - \phi \Omega \big)(1-\Omega)^{-1},$
 where $\phi$ is the maximum output DC power, $\Omega=1/(1 + \exp(\xi \chi) )$ is a constant to guarantee a zero input/output response, while $\xi$ and $ \chi$ are constant related parameters that depend on the circuit. Thus, we obtain the energy at the output of the circuit as a logistic function of the received input RF energy at ER $k_e$, i.e.: $\Lambda(\mathcal{E}_{k_e}) = \phi \big(1 + \exp\big(\!-\xi\big(\mathcal{E}_{k_e}\!-\! \chi\big)\!\big) \big)^{-1}$.
To this end, we consider the average of the total HE at ER $k_e$ as the performance metric of the WPT operation, which is
 \vspace{-0.1em}
  \begin{align}~\label{eq:NLEH:av}
  \Ex\big\{\mathcal{E}^{\mathrm{NL}}_{k_e}\big\} &= \frac{\Ex\big\{\Lambda\big(\mathcal{E}_{k_e}\big)\big\} - \phi \Omega }{1-\Omega}
  \approx
  \frac{\Lambda\big(\mathcal{Q}_{k_e}\big) - \phi \Omega }{1-\Omega},
 \end{align}
where we have used the approximation $\Ex\left\{\Lambda\left(\mathrm{E}_{k_e}\right)\right\} \approx\Lambda\left(Q_{k_e}\big\}\right)$~\cite{Mohammadi:TC:2024}. In~\eqref{eq:NLEH:av}, $\mathcal{Q}_{k_e} \triangleq\Ex\left\{\mathcal{E}_{k_e}\right\}$ is the average of the received RF energy at the ER $k_e$, given by
\begin{align}~\label{eq:El_average}
     &\mathcal{Q}_{k_e}
     =(\tau_c-\tau)\Snn
     \Big(\! \snrdl \!\!
     \sum\nolimits_{m\in\M_{\mathtt{E}}}
   {\etamjE} \Ex\Big\{\!\big\vert\big(
   \gejue
   \big)^\dag \wemj^{\MRT}\big\vert^2\!\Big\} \nonumber \\
   &\hspace{=0.2em}+
    {\rho}_d \!\!
    \sum\nolimits_{k'_e \in\K_{\mathtt{E}} \!\! \setminus k_e}\!
    \sum\nolimits_{m\in\M_{\mathtt{E}}} \!\!
   {\etamjpE} \Ex\Big\{\! \big\vert\big(\gejue\big)^\dag\wemjp^{\MRT}\big\vert^2\!\Big\}
   \nonumber\\
   &\hspace{-0.4em} +\!{\rho}_d \!\!
   \sum\nolimits_{k_i\in\K_{\mathtt{I}}}\!
   \sum\nolimits_{m\in\M_{\mathtt{I}} }\!\!\!
   {\etamkI}\Ex\Big\{\big\vert\big(\gejue\big)^{\!\dag}\wimk^{\MRT}\big\vert^2\!\Big\} 
   \!+\! 1
     \!\Big).
\end{align}

\begin{proposition}~\label{Theorem:SE:MRT}
With MRT precoding, the achievable SE of the $k_i$-th IR is given by~\eqref{eq:SEk:Ex}, where the effective SINR is 
\vspace{-0.1em}
\begin{align}~\label{eq:SINE:MRT}
    &\SINRki^{\MRT}\big(\boldsymbol{\Theta}, \ETAI, \ETAE\big) = \\
    &\!\frac{
            \Big( \sum\nolimits_{m \in \M_{\mathtt{I}}} \sqrt{ \rho_d \etamkI \alpha^{\MRT}_{m k_i} }  \Big)^2
            }
            { 
            \mathsf{BU}_{k_i}
            +
            \mathsf{IUI}_{k_i, \mathrm{PC}}
                        +
            \mathsf{EUI}_{k_i, \mathrm{PC}}
            +
            \mathsf{IUI}_{k_i, \mathrm{O}}
            +
            \mathsf{EUI}_{k_i, \mathrm{O}}
            +  1/\snrdl
            },\nonumber
\end{align}
where 
\begin{align}
\mathsf{BU}_{k_i} \!\!  \triangleq&\! \sum\nolimits_{m \in \M_{\mathtt{I}}} \! \snrdl \etamkI \Big( \frac{\bar{a}_{m k_i}}{\alpha^{\MRT}_{m k_i}} \nonumber\\
&+
    \barbetamki \trace(\FmRIS \FmRIS^{\dag}) - \gamma^{\mathtt{I}}_{m k_i} - \alpha^{\MRT}_{m k_i} \Big),\nonumber\\
      \mathsf{IUI}_{k_i, \mathrm{PC}} \! \triangleq&  
    \sum\nolimits_{k'_i \in \mathcal{P}_k \setminus k_i}
    \sum\nolimits_{m \in \M_{\mathtt{I}}} \! \snrdl \etamkpI 
    \nonumber\\
    & \times \Big( \frac{\bar{b}_{m, k_i k'_i}}{\alpha^{\MRT}_{m k'_i} } + \barbetamki \trace(\FmRIS \FmRIS^{\dag} ) - \gamma^{\mathtt{I}}_{m k_i} \Big), \nonumber
\end{align}
\begin{align}
    \mathsf{EUI}_{k_i, \mathrm{PC}}   \triangleq& \! 
    \sum\nolimits_{k_e \in \mathcal{P}_k }
    \sum\nolimits_{m \in \M_{\mathtt{E}}} \!  \snrdl \etamjE 
    \nonumber\\
    & \times
    \Big( \frac{\bar{b}_{m, k_i k_e}}{\alpha^{\MRT}_{m k_e} } + \barbetamki \trace(\FmRIS \FmRIS^{\dag} ) - \gamma^{I}_{m k_i} \Big), \nonumber
\end{align}
where $\bar{b}_{m, k_i k'_i}$ and $\bar{b}_{m, k_i k_e}$ are defined in Lemma~\ref{lemma2} (see Appendix A). Moreover,
\begin{align}
    \mathsf{IUI}_{k_i, \mathrm{O}} \!  &= \!
    \sum\nolimits_{k'_i \notin \mathcal{P}_k \setminus k_i}
    \sum\nolimits_{m \in \M_{\mathtt{I}}} \! \snrdl \etamkpI \barbetamki
    \nonumber\\
    & \times \Big(  \trace\big(\FmRIS \FmRIS^{\dag}) + 
    \kappa \trace(\FmRIS \FmRIS^{\dag} \zmkilos \zmkilos^{\dag} \big)  \Big), \nonumber
\end{align}
\begin{align}
    \mathsf{EUI}_{k_e, \mathrm{O}} \!  &= \! 
    \sum\nolimits_{k_e \notin \mathcal{P}_k }
    \sum\nolimits_{m \in \M_{\mathtt{E}}} \! \snrdl \etamjE \barbetamki
    \nonumber\\
    & \times \Big(  \trace\big(\FmRIS \FmRIS^{\dag}) + 
    \kappa \trace(\FmRIS \FmRIS^{\dag} \zmkilos \zmkilos^{\dag} \big)  \Big).\nonumber
\end{align}
\end{proposition}
\begin{proof}
    See Appendix \ref{appendix:B}.
\end{proof}

\begin{proposition}~\label{Theorem:RF:MRT}
With MRT precoding, the average HE at  ER $k_e \in \K_{\mathtt{E}}$ is given by~\eqref{eq:NLEH:av}, where $\mathcal{Q}_{k_e}^{\MRT}$ is given by
\begin{align}~\label{eq:El_average:MRT}
    &\mathcal{Q}_{k_e}^{\MRT} =
    (\tau_c - \tau)\Snn \snrdl
    \Big[ 1/ \snrdl +
    \sum\nolimits_{m \in \M_{\mathtt{E}}} 
    \etamjE  \bar{\mu}_{m k_e}
    \nonumber\\
    &\hspace{-0.5em}
    \!+\!
    \sum\nolimits_{m \in \M_{\mathtt{E}}} \!\!
    \Big( \!
   \sum\nolimits_{k'_e \in \mathcal{P}_k \setminus k_e} \!\!
     \etamjpE  
   \bar{\epsilon}_{m, k_e k'_e}
    \!+\!
   \sum\nolimits_{k'_e \notin \mathcal{P}_k \setminus k_e}
   \etamjpE \bar{c}_{m k_e}
    \!\! \Big)
    \nonumber\\
    &\hspace{-0.5em}
    \!+\!
    \sum\nolimits_{m \in \M_{\mathtt{I}}} \!\!
    \Big( \!
    \sum\nolimits_{k_i \in \mathcal{P}_{k}}
    \etamkI    \bar{\epsilon}_{m, k_e k_i}
    \!\!+\!\!
    \sum\nolimits_{k_i \notin \mathcal{P}_{k}} \!\!
    \etamkI    \bar{c}_{m k_e}
    \!\! \Big) 
    \Big],
\end{align}
where $\bar{c}_{m k} \triangleq \barbetamk \big( 
    \trace(\FmRIS \FmRIS^{\dag})
    +
    \kappa \trace(\FmRIS \FmRIS^{\dag} \zmklos \zmklos^{\dag})
    \big)$, $ \bar{\mu}_{m k} \triangleq \frac{\bar{a}_{m k}}{\alpha^{\MRT}_{m k}}  + \barbetamke \trace(\FmRIS \FmRIS^{\dag}) - \gamma^{\mathtt{i}}_{m k} $, and $\bar{\epsilon}_{m, k k'} \triangleq \frac{\bar{b}_{m, k k'}}{\alpha^{\MRT}_{m k}} + \barbetamke \trace(\FmRIS \FmRIS^{\dag}) -  \gamma^{\mathtt{i}}_{m k}$.

\begin{proof}
The proof follows by deriving the expectation terms in~\eqref{eq:El_average}. To this end, we exploit Lemma~1 and Lemma~2 in Appendix A and omit the detailed derivations due to space limitations.
\end{proof}
\end{proposition}
\vspace{-1.6em}
\subsection{Problem Formulation}~\label{ref:ProblemFormulation}
We formulate an optimization problem to maximize the average sum-HE, subject to the minimum HE required per ER $\Gamma_{k_e}$, the SE threshold per IR $\mathcal{S}_{k_i}$, and the power budget for each AP.
Our objective is to determine the power control coefficients, $\ETAI, \ETAE$, and the SIM's PS $\boldsymbol{\Theta}(\theta^{m \ell}_{s}) \triangleq \{ \boldsymbol{\Psi}^{m}(\theta^{m \ell}_{s}), \forall m \in \M \}$. It is worth noting that the system variable $\boldsymbol{\Theta}$ is a function of the PS elements $\theta^{m \ell}_{s}$, which is defined in Section~\ref{sec:Sysmodel}. The optimization problem can be mathematically expressed as
\begin{subequations}~\label{eq:ProblemFormulationorigin}
    \begin{align}
        (\mathcal{P}1): & \quad \max_{\ETAI, \ETAE, \boldsymbol{\Theta} } \quad \sum\nolimits_{k_e \in \K_{\mathtt{E}}}{\Ex\big\{\mathcal{E}^{\mathrm{NL}}_{k_e}\big\}}~\label{obj:ObjectiveFunction1}
        \\
        \mathrm{s.t.} 
        &\hspace{1em} \Ex\big\{\mathcal{E} ^{\mathrm{NL}}_{k_e}\big\} \geq \Gamma_{k_e}, \forall k_e \in \K_{\mathtt{E}},~\label{ct:HE_threshold}
        \\
        &\hspace{1em} \text{SE}_{k_i} \geq \mathcal{S}_{k_i}, \forall k_i \in \K_{\mathtt{I}} ,~\label{ct:SINRThreshol}
        \\
        &\quad \sum\nolimits_{k_i \in \K_{\mathtt{I}}}{\etamkI} \leq 1, \forall m \in \M \setminus e,~\label{ct:etamkI2}
        \\
        &\quad \sum\nolimits_{k_e\in\K_{\mathtt{E}}}{\etamjE} \leq 1,~\label{ct:etamjE2}
        \\
        &\hspace{1em} 
        \theta^{m \ell}_{s} \in [0, 2 \pi], \forall m \in \M, \forall \ell \in \mathcal{L}, \forall s \in \mathcal{S}.
    \end{align}
\end{subequations}
where \eqref{ct:etamkI2} and \eqref{ct:etamjE2} represent the transmission power constraints at the APs, while $\Gamma_{k_e}$ and $\mathcal{S}_{k_i}$ are the minimum QoS requirements at the ER $k_e$ and IU $k_i$.

\vspace{-0.8em}
\section{Proposed Solution}~\label{sec:SCAOptimization}
In this section, we present an efficient solution to the optimization problem ($\mathcal{P}1$). Since the joint optimization problem ($\mathcal{P}1$) is challenging, we divide it into two sub-problems: i) PS design; ii) Power allocation design. In Section~\ref{subsec:heu} we introduce an iterative approach for designing a suboptimal $\boldsymbol{\Theta}$, which prioritizes the WPT operation. Subsequently, in Section~\ref{subsec:ppa} we address the power allocation sub-problem.
\vspace{-1.0em}
\subsection{Heuristic Phase Shift Design (\textbf{Heuristic-PS})}~\label{subsec:heu}
\vspace{-1.4em}
\begin{algorithm}[t]
\caption{Layer-by-layer heuristic search for $L$-layer SIM}~\label{alg:layerbylayerSIM}
    \begin{algorithmic}[1]
    \small
        \STATE Initiate SIM with stochastic ${\Phi}^{m, \ell}$ in range of $[0, 2\pi)$, $\forall \ell \in \mathcal{L}, \forall m \in \M$
        \STATE \textbf{for} $m = 1$ \textbf{to} $M$:
        \STATE \quad \textbf{for} $l = 1$ \textbf{to} $L$:
        \STATE \quad \quad \textbf{for} $\mathtt{it} = 1$ \textbf{to} $C$:
            \STATE  \quad \quad \quad 
            Generate $\boldsymbol{\Phi}^{m \ell}$ and compute $f(\mathtt{it})$ according to~\eqref{eq:heu:obj}
            \STATE  \quad \quad \quad 
            Store $\boldsymbol{\Phi}^{m \ell}(\mathtt{it}), \mathrm{obj}(\mathtt{it})$ and its index
            \STATE \quad \quad 
            Search the maximum $\mathrm{obj}(\mathtt{it}^*)$ value and its index $\mathtt{it}^*$
            \STATE \quad \quad
            Assign $\mathtt{it}^*$-th $\boldsymbol{\Phi}^{m \ell}$ as optimal $(\boldsymbol{\Phi}^{m \ell})^{*}$
            \STATE \quad \quad Move to $\ell + 1$ layer until the last layer
            \STATE \quad  Move to $m + 1$ AP until the last AP
        \RETURN $(\boldsymbol{\Theta}^{m})^{*} \triangleq \{ (\boldsymbol{\Phi}^{m \ell})^{*}, \forall m \in \M, \ell \in \mathcal{L} \}$.
    \end{algorithmic}
\end{algorithm}
\setlength{\textfloatsep}{0.1cm}

From~\eqref{eq:El_average:MRT}, we notice that the SIM's PS directly affects the trace value of the SIM-affected channel with its Hermitian, which serves as a dominant term in the expression. Thus, under $C$ number of cyclic network realizations, we heuristically optimize the PSs by solving the problem:
\begin{equation}~\label{eq:heu:obj}
    \max_{\boldsymbol{\Theta}^{m}}{f(\boldsymbol{\Theta}) = \trace(\FmRIS \FmRIS^{\dag})}.
\end{equation}
\textbf{Algorithm}~\ref{alg:layerbylayerSIM} provides the layer-by-layer heuristic scheme while the \textbf{complexity} of the heuristic scheme is $\mathcal{O}(MCSL^{2} )$. Notably, the effectiveness of the heuristic search is proportional to $C$, as a higher number of search allows for extensive explorations of the solution space. Thus, $C$ serves as a trade-off parameter between computational efficiency and optimization performance.

\subsection{Proposed Power Allocation (\textbf{PPA}) Design}~\label{subsec:ppa}
Before proceeding, by invoking~\eqref{eq:El_average} and after some manipulations, we can replace constraint \eqref{ct:HE_threshold} by $\mathcal{Q}_{k_e}^{\MRT} \geq \Xi(\tilde{\Gamma}_{k_e}),~\forall k_e \in \K_{\mathtt{E}},$
where $\tilde{\Gamma}_{k_e} = (1-\Omega) \Gamma_{k_e} + \phi\Omega$ and $\Xi(\tilde{\Gamma}_{k_e}) = \chi - \frac{1}{\xi}\Big( \frac{\phi -\tilde{\Gamma}_{k_e}}{\tilde{\Gamma}_{k_e}} \Big)$ is the inverse function of the logistic function.  
Subsequently, we denote $\boldsymbol{\epsilon} =\{\epsilon_{k_e}\geq 0, k_e\in\K_{\mathtt{E}}\}$ as auxiliary variables, $\tilde{\epsilon}_{k_e} \triangleq (1 - \Omega) \epsilon_{k_e} + \phi \Omega$.
Given $\boldsymbol{\Theta}^{*}$, the sum-HE problem~\eqref{eq:ProblemFormulationorigin} can be recast as
\begin{subequations}~\label{eq:ProblemFormulationP2}
    \begin{align}
        (\mathcal{P}2.1): & \quad \max_{\ETAI, \ETAE, \boldsymbol{\epsilon}} \quad \sum\nolimits_{k_e \in \K_{\mathtt{E}}}{\epsilon_{k_e}}~\label{obj:ObjectiveFunction2}
        \\ 
        \mathrm{s.t.} 
        &\quad 
        \mathcal{Q}_{k_e}^{\MRT} \geq \Xi(\tilde{\epsilon}_{k_e}),~\forall k_e \in \K_{\mathtt{E}},~\label{ct:auxilaryvariable2a}      
        \\
        &\quad \epsilon_{k_e} \geq \Gamma_{k_e}, \forall k_e \in \K_{\mathtt{E}},~\label{ct:auxilaryvariable2b2}
        \\
        &~\eqref{ct:SINRThreshol}-\eqref{ct:etamjE2}.~\label{ct:etamjE22}
    \end{align}
\end{subequations}
Note that the non-convex nature of~\eqref{ct:auxilaryvariable2a} and~\eqref{ct:SINRThreshol} leads to the non-convexity of the problem~\eqref{eq:ProblemFormulationP2}. To address the non-convex constraint~\eqref{ct:auxilaryvariable2a}, we use the convex upper bound of $\Xi(\tilde{\epsilon}_{k_e})$ as
\begin{align*}
        \Xi(\tilde{\epsilon}_{k_e}) &\leq   \tilde{\Xi}(\tilde{\epsilon}_{k_e})\triangleq \chi - \frac{1}{\xi} \bigg( \ln\Big(\frac{\phi - \tilde{\epsilon}_{k_e}}{\tilde{\epsilon}_{k_e}^{(n)}} \Big)
        - \frac{\tilde{\epsilon}_{k_e} - \tilde{\epsilon}_{k_e}^{(n)}}{\tilde{\epsilon}_{k_e}^{(n)}} \bigg),
\end{align*}
where the superscript ($n$) denotes the value of the involving variable produced after $(n - 1)$ iterations ($n \geq 0$). Thus, we obtain the convex approximation of~\eqref{ct:auxilaryvariable2a}  as
\begin{align}~\label{ct:auxilaryvariable2a_convex}
    &1 \!+  \!
    \snrdl
    \sum\nolimits_{m \in \M_{\mathtt{E}}}
    \Big(
    \etamjE \bar{\mu}_{m k_e}
     \!+ \!
    {\snrdl}\sum\nolimits_{k'_e \in \K_{\mathtt{E}} \setminus k_e}   \etamjpE  
    \bar{\epsilon}_{m, k_e k'_e}
    \Big)
    \nonumber\\
    &+\!
    \snrdl
    \sum\nolimits_{k_i \in \K_{\mathtt{I}}}
    \!
    \sum\nolimits_{m \in \M_{\mathtt{I}}} 
    \!
    \etamkI    \bar{c}_{m k_e}
    \!\geq\! \tilde{\Xi}(\tilde{e}_{k_e}) / (\tau_c \!-\! \tau )\Snn.
\end{align}

Now we focus on~\eqref{ct:SINRThreshol}, which is equivalent to
\begin{align}~\label{eq:ctSINR1}
    &\frac{1}{\mathcal{T}_{k_i}} 
    \Big( 
    \sum\nolimits_{m \in \M_{\mathtt{I}}} 
    \!\! \sqrt{ \snrdl \etamkI \alpha^{\MRT}_{m k_i} }  \Big)^2 
    \!\geq\!
    \sum\nolimits_{m \in \M_{\mathtt{I}}}  \!\!\snrdl \etamkI \bar{\varrho}_{m k_i}
    \nonumber\\
    &
    \sum\nolimits_{m \in \M_{\mathtt{I}}} \!\!
    \Big( \!
    \sum\nolimits_{k_{i}'\in\K_{\mathtt{I}}\!\! \setminus k_{i}} 
    \snrdl \etamkpI \bar{c}_{m k_i}
    \!+\!\!
    \sum\nolimits_{k_{e}\in\K_{\mathtt{E}} } 
    \snrdl \etamjE \bar{c}_{m, k_i} \! \Big)
    \nonumber\\
    &+ 1,
\end{align}
where $\bar{\varrho}_{m k_i} \triangleq \frac{\bar{a}_{m k_i}}{\alpha^{\MRT}_{m k_i}} +
    \barbetamki \trace(\FmRIS \FmRIS^{\dag}) - \gamma^{\mathtt{I}}_{m k_i} - \alpha^{\MRT}_{m k_i}$.

To deal with the non-convexity in~\eqref{eq:ctSINR1}, we apply the following concave lower bound $x^2 \geq x^{(n)}(2x - x^{(n)})$ to~\eqref{eq:ctSINR1} as
\vspace{-0.5em}
\begin{align}~\label{eq:ctSINR2}
    &\frac{ q_{k_i}^{(n)}}{\mathcal{T}_{k_i}} \big( 2 \! \sum\nolimits_{m \in \M_{\mathtt{I}}} \sqrt{ \snrdl \etamkI \alpha^{\MRT}_{m k_i} } - q_{k_i}^{(n)} \big)
    \nonumber\\
    &\geq
    \snrdl
    \sum\nolimits_{m \in \M_{\mathtt{I}}} 
    \big(
    \etamkI \bar{\varrho}_{m k_i}
    +
    \sum\nolimits_{k_{e}\in\K_{\mathtt{E}} } 
     \etamjE \bar{c}_{m, k_i} 
     \big)
    \nonumber\\
    &+
    \sum\nolimits_{m \in \M_{\mathtt{I}}} 
    \sum\nolimits_{k_{i}'\in\K_{\mathtt{I}} \setminus k_{i}} 
    \snrdl \etamkpI \bar{c}_{m k_i}
    + 1,
\end{align}
where $q_{k_i} \triangleq \sum\nolimits_{m \in \M_{\mathtt{I}}} \sqrt{\snrdl \etamkI \alpha^{\MRT}_{m k_i} }$. Therefore, the optimization problem~\eqref{eq:ProblemFormulationP2} is now recast as
\vspace{-0.2em}
\begin{subequations}\label{opt:JAP:final}
\begin{alignat}{2}
&(\mathcal{P}2.2):  \max_{\ETAI, \ETAE, \boldsymbol{\epsilon}}        
&\qquad&\sum\nolimits_{k_e \in \K_{\mathtt{E}}}{\epsilon_{k_e}}~\label{opt:JAP:final:obj}\\
&\hspace{5em}\text{s.t.} 
&         &~\eqref{ct:auxilaryvariable2a_convex},~\forall k_e \in \K_{\mathtt{E}},~\label{opt:JAP:final:ct1}\\
&         &      &~\eqref{eq:ctSINR2},~\forall k_i \in \K_{\mathtt{I}},
           ~\label{opt:JAP:final:ct2}\\
&         &      &~\eqref{ct:auxilaryvariable2b2},~\eqref{ct:etamjE2},~\eqref{ct:etamkI2}.~\label{opt:JAP:final:ct3}
\end{alignat}
\end{subequations}
Problem~\eqref{opt:JAP:final} is now convex and thus it can be rebustly solved using CVX~\cite{cite:Grant:CVX}. \textbf{Algorithm~\ref{alg1}} outlines the main steps to solve problem ($\mathcal{P}2.2$), where $\widetilde{\qx} \triangleq \{\ETAI, \ETAE, \boldsymbol{\epsilon}\}$ and $\widehat{\mathcal{F}} \triangleq\{\eqref{opt:JAP:final:ct1},~\eqref{opt:JAP:final:ct2},~\eqref{opt:JAP:final:ct3}\}$  is a convex feasible set. Starting from a random point $\widetilde{\qx}\in\widehat{\mathcal{F}}$, we solve \eqref{opt:JAP:final} to obtain its optimal solution $\widetilde{\qx}^*$, and use $\widetilde{\qx}^*$ as an initial point in the next iteration.

\textbf{Complexity analysis:} In each iteration of \textbf{Algorithm~\ref{alg1}}, the computational complexity of solving  \eqref{opt:JAP:final} is $\OO(\sqrt{A_l}(A_v+A_l)A_v^2)$ since~\eqref{opt:JAP:final} can be equivalently reformulated as
an optimization problem with  $A_v\triangleq M(K_i+2K_e)$ real-valued scalar variables, $A_l\triangleq 2K_e+K_i+M$ linear constraints.

\begin{algorithm}[t]
\caption{Proposed power allocation (\textbf{PPA}) algorithm }
\begin{algorithmic}[1]
\label{alg1}
\STATE \textbf{Initialize}: $n\!=\!0$, 
$\lambda^{\mathrm{pen}} > 1$, random initial point $\widetilde{\qx}^{(0)}\!\in\!\widehat{\mathcal{F}}$.
\REPEAT
\STATE Update $n=n+1$
\STATE Solve \eqref{opt:JAP:final} to obtain its optimal solution $\widetilde{\qx}^*$
\STATE Update $\widetilde{\qx}^{(n)}=\widetilde{\qx}^*$
\UNTIL{convergence}
\end{algorithmic}
\end{algorithm}
\setlength{\textfloatsep}{0.0cm}

\vspace{-0.5em}
\section{Numerical Results}~\label{sec:PerformanceAnalysis}
We present numerical results to evaluate the performance of the SIM-assisted CF-mMIMO SWIPT system. We consider that the APs, $K_i=3$ IRs and $K_e=4$ ERs are randomly distributed in a square of $0.1 \times 0.1$ km${}^2$. The height of the SIM-attached APs and the receivers are $15$~m, and 1.65~m, respectively. We set $\tau_c = 200$, $\Snn = -92$ dBm, $\rho_d = 1$~W and $\rho_u = 0.2$ W.  In addition, we set the NL-EH molde's parameters as $\xi = 150$, $\chi = 0.024$, and $\phi = 0.024$~W~\cite{Mohammadi:TC:2024, Hua:WCNC:2024}. The large-scale fading coefficients are generated following the three-slope propagation model from \cite{cite:HienNgo:cf01:2017}. Shadow fading follows a log-normal distribution with a standard deviation of $8$ dB. Regarding the SIM architecture, we consider that the spacing between two adjacent metasurfaces is $d^{\ell}_{s, \breve{s}} = 5\lambda/L$ and $d_{\mathrm{PS}}=d_{\mathrm{SIM}}= \lambda/2$, respectively. We consider the SIM's PS configuration benchmark as \textbf{Random-PS} and \textbf{Equal-PS}~\cite{cite:Chien:TWC:2022}, which sets stochastic PS and equal PS per layer, respectively.  Moreover, we consider equal power allocation (\textbf{EPA}) to evaluate the performance of the \textbf{PPA} design.
\begin{figure*}[t]
\vspace{0.1cm}
    \centering
    \begin{minipage}[t]{0.32\textwidth}
        \centering
        \includegraphics[trim=0 0cm 0cm 0cm,clip,width=1.05\textwidth]{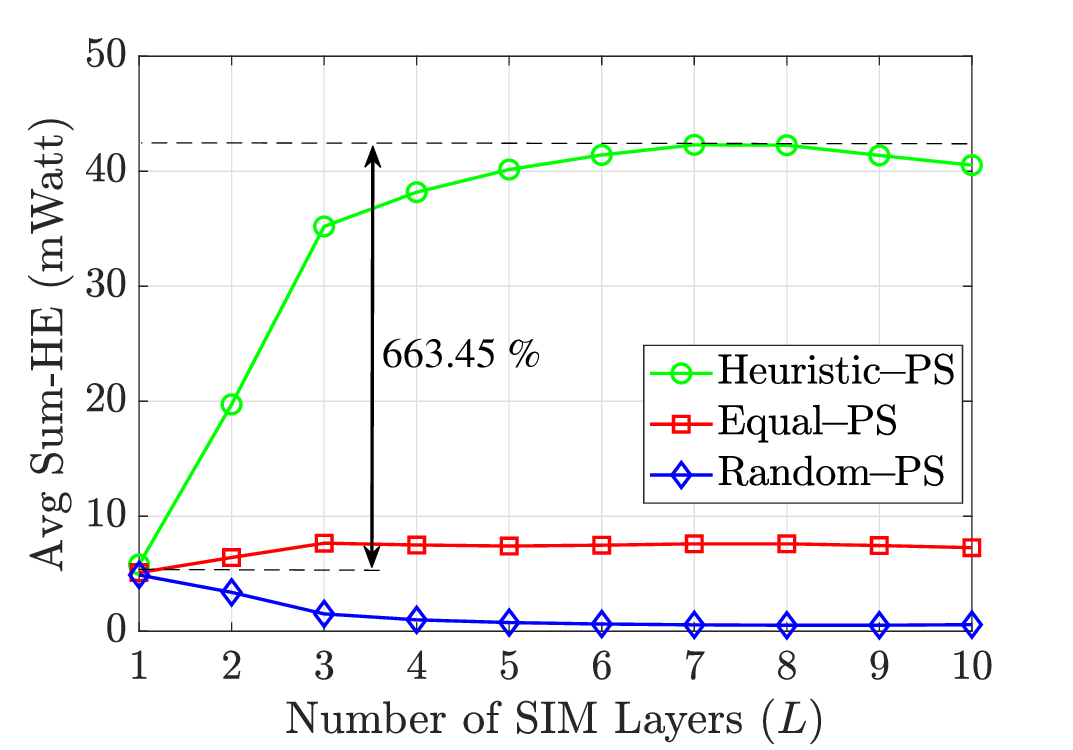}
        \vspace{-1.7em}
        \caption{\small Impact of SIM layers' number on the sum-HE ($M = 10$, $N = 20$, $S = 25$).\normalsize}
    \label{fig:sumHE_v_simlayers}
    \end{minipage}
    \hfill
    \begin{minipage}[t]{0.32\textwidth}
        \centering
        \includegraphics[trim=0 0cm 0cm 0cm,clip,width=1.05\textwidth]{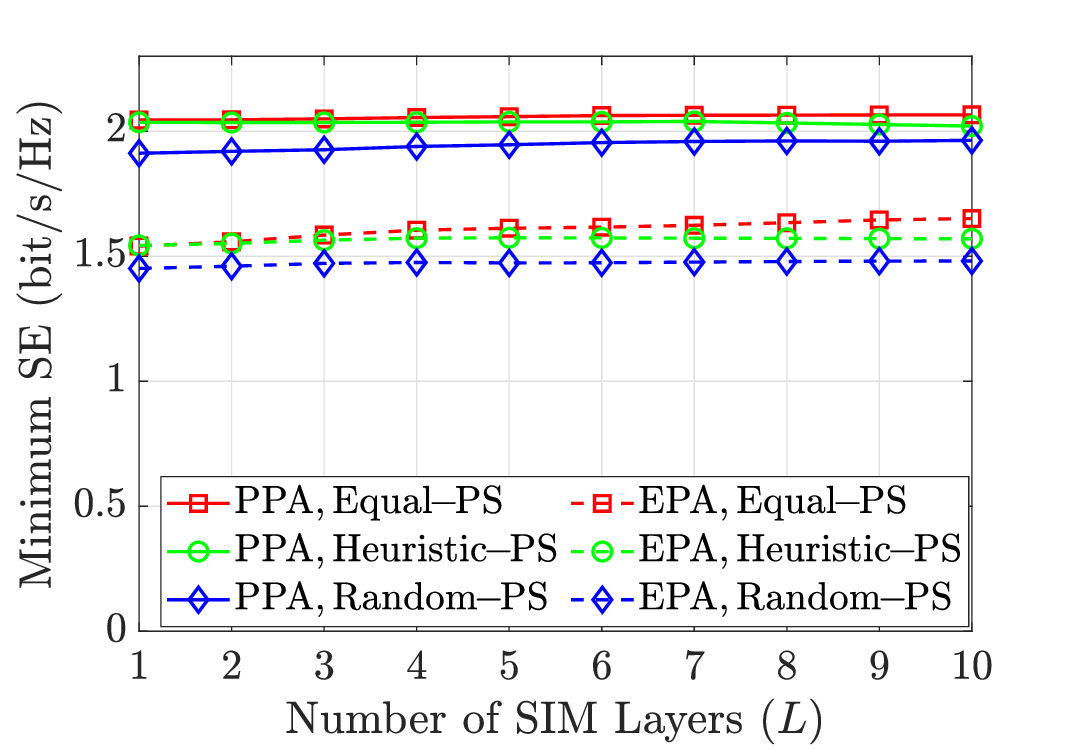}
        \vspace{-1.7em}
        \caption{\small Impact of SIM layers' number on the minimum SE ($M = 10$, $N = 20$, $S = 25$).\normalsize}
        \label{fig:minSE_v_simlayers}
    \end{minipage}
    \hfill
    \begin{minipage}[t]{0.32\textwidth}
        \centering
        \includegraphics[trim=0 0cm 0cm 0cm,clip,width=1.05\textwidth]{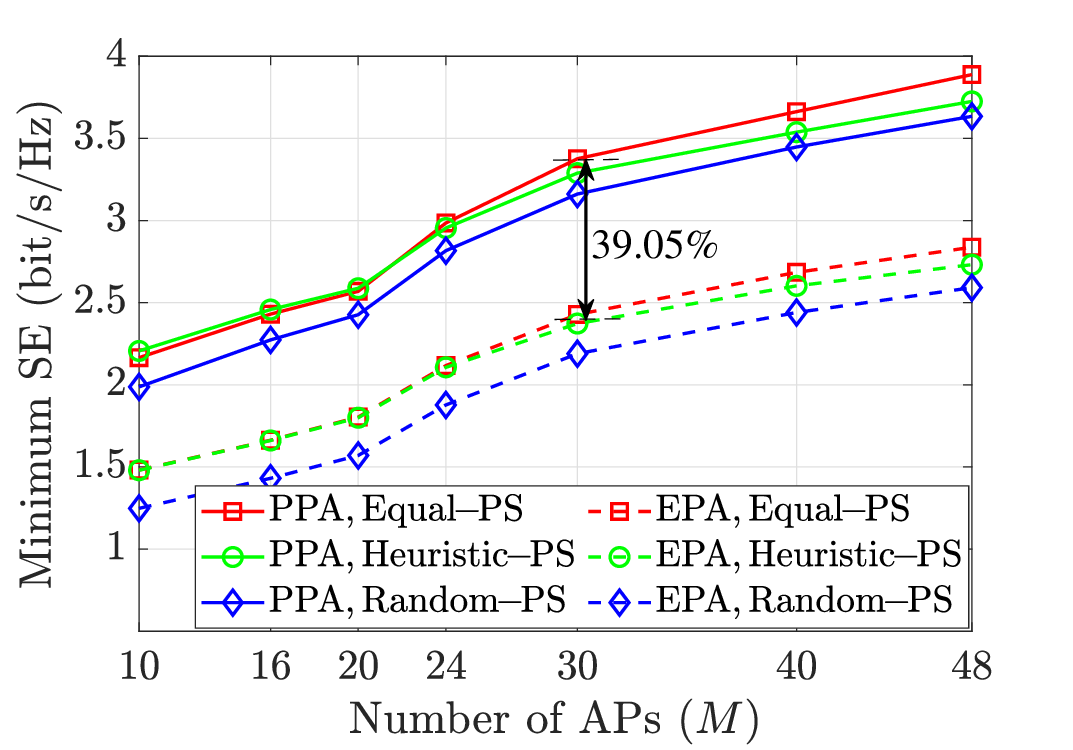}
        \vspace{-1.7em}
        \caption{\small Impact of the \textbf{PPA} design on the minimum SE ($MN = 480$, $S = 25$, $L = 4$).\normalsize}
        \label{fig:minSE_v_M_SCA}
    \end{minipage}
\vspace{-1.4em}
\end{figure*}

Figure~\ref{fig:sumHE_v_simlayers} shows the impact of the number of SIM layers on the average sum-HE. We first notice a $6$-fold enhancement in the sum-HE gain by increasing the number of SIM layers up to $7$ layers per SIM. Subsequently, since the objective function~\eqref{eq:heu:obj} favors the dominant term in~$\mathcal{Q}_{k_e}^{\MRT}$, we observe a 456.58\% average sum-HE gain of the~\textbf{Heuristic-PS} compared to the \textbf{Equal-PS}. As $L$ increases, the \textbf{Heuristic-PS} enhances the WPT performance significantly up to a certain value of $L$, while \textbf{Random-PS} degrades the HE performance. This can be attributed to that configuring random PS per layer causes destructive interference and signal scattering, leading to a reduction in $\Vert\FmRIS\Vert^{2}_{\mathrm{F}}$. However, as $L\geq8$,  the geometric distribution of layers in~\eqref{eq:Fm} begins to overlap, causing the beamforming to become excessively dispersed or scattered, which results in a reduction in the sum-HE metric. Thus, an appropriate selection of SIMs layers must be considered under specific scenarios.

Figure~\ref{fig:minSE_v_simlayers} and~\ref{fig:minSE_v_M_SCA} demonstrate the performance of the \textbf{PPA} on the minimum SE under the influence of the number of SIM layers and number of APs, respectively.
As explained in Section~\ref{sec:SCAOptimization}, iteratively optimizing the SIM's PS improves the WPT operation. Additionally, the achievable minimum SE remains unaffected and increases slightly by approximately $7.15\%$. This highlights the effectiveness of our proposed SIMs-assisted CF-mMIMO SWIPT system model and justifies the implementation of the \textbf{PPA} scheme to enhance the WIT operation. We note that, for a fixed number of service antennas (i.e., $MN = 480$), the number of antennas per AP decreases as the number of APs increases. As $M$ increases, the minimum SE shows a significant improvement, with a gain of approximately $77.27\%$ across all SIM's PS schemes. Regarding the proposed optimization design, we observe a $39\%$ increase in the minimum SE achieved by the \textbf{PPA} scheme compared to the \textbf{EPA} scheme for all SIM's PS configurations.

\vspace{-0.3em}
\section{Conclusion}~\label{sec:conclusion}
We analyzed the SE and average HE performance in a SIM-assisted CF-mMIMO SWIPT setup, where the APs are equipped with a SIM and simultaneously perform two distinct WIT and WPT operations. By leveraging closed-form expressions for the average sum-HE of ERs and the achievable SE of IRs, we formulated and solved a non-convex optimization problem to maximize sum-HE. Our simulation results showed an up to a $7$-fold gain in the sum-HE with the proposed layer-by-layer heuristic SIM's PS design over the random PS design. Moreover, a $44\%$ improvement in the achievable minimum SE can be achieved with the \textbf{PPA} scheme compared to the \textbf{EPA}.

\appendices
\vspace{-0.4em}
\section{Useful Derivations for Proposition 2}\label{appendix:usefulDerivations}
Given the complex nature of the aggregated channel, we provide the following Lemmas to facilitate our derivations.
\vspace{-0.5em}
\begin{Lemma}~\label{lemma1}
    Let $\bar{a}_{m k}$ denote the fourth-order moment of the estimate channel $\hgmk$, where $\hgmk \sim \CN\Big(\gmkbar, \Covhatgmk \Big)$. We can then obtain
    \begin{align}~\label{eq:4thmomenthatgmk}
        \bar{a}_{m k} &\triangleq \Ex\{\Vert \hgmk \Vert^4 \} 
        = \Ex\big\{ \big\vert (\hgmk)^\dag  \hgmk  \big\vert^2 \big\} 
        \nonumber\\
        &\stackrel{(a)}{=}
        \big\vert \trace\big(\Covhatgmk \big) \big\vert^{2} + \big\vert \trace\big(\Covhatgmk \Covhatgmk \big) \big\vert
        +
        \big\vert (\gmkbar)^\dag \gmkbar \big\vert^2
        \nonumber\\
        &+
        2 (\gmkbar)^\dag \gmkbar \trace\big(\Covhatgmk \big)
        +
        2 (\gmkbar)^\dag \Covhatgmk \gmkbar, 
    \end{align}
where (a) exploits \cite[Chapter 1, Lemma 9]{Chien2020MassiveMC}.

\end{Lemma}

\begin{Lemma}~\label{lemma2}
    Let $\bar{b}_{m, k k'}$ denote the fourth-order moment of the estimate channel $\hgmk$, where $\hgmk \sim \CN\big(\gmkbar, \Covhatgmk \big)$. Sharing the same pilot sequence with $\hgmkp$, we obtain
    \vspace{0.5em}
    \begin{align}~\label{eq:4thmomenthatgmk_PC}
        \bar{b}_{m, k k'} &\triangleq \Ex\big\{ \big\vert (\hgmk)^\dag  \hgmkp  \big\vert^2 \big\}
        \nonumber\\
        &=
        \vert (\gmkbar)^{\dag} \gmkpbar \vert^2 
        +
        2 \upsilon_{m,k,k'} (\gmkbar)^{\dag} \gmkpbar \trace\big(\Covhatgmk \big)
        \nonumber\\
        &+
        (\upsilon_{m,k,k'})^2 \Big( \big\vert \trace\big(  \Covhatgmk \big) \big\vert^{2} 
        + \trace\big(\Covhatgmk \Covhatgmk \big) \Big)
        \nonumber\\
        &+
        2 \upsilon_{m,k,k'} (\gmkbar)^{\dag} \Covhatgmk \gmkpbar
        , \quad \forall k, k' \in \mathcal{P}_{k}.
    \end{align}
    where $\upsilon_{m,k,k'} \triangleq (\barbetamkp)^2 / (\barbetamk)^2$.
\end{Lemma}

\vspace{-1.7em}
\section{Proof of Proposition 1}~\label{appendix:B}
1) \textit{Compute $\DSki$}: Given that $\hgmk$ and $\tilgmk$ are independent due to the properties of MMSE estimation, we obtain 
\begin{align}~\label{eq:DS_k}
\mathsf{DS}_k &= \sum\nolimits_{m \in \M_\mathtt{I}} \sqrt{\rho_d \etamkI} \Ex \big\{(\hgmk + \tilgmk )^{H} \wimk^{\MRT} \big\} \nonumber \\
&= \sum\nolimits_{m \in \M_\mathtt{I}}{ \sqrt{\alpha^{\MRT}_{mk_i}\rho_d \etamkI}}.
\end{align}

2) \textit{Compute $\mathsf{BU}_k$}: Invoking that the variance of a sum of independent RVs is equal to the sum of the variances, we obtain
\begin{align}~\label{eq:BU_k}
    & \sum\nolimits_{m \in \M_\mathtt{I}} \! \snrdl \etamkI  \! \Ex \Big\{ \Big\vert (\gmkiu)^{\dag} \wimk^{\MRT} \!-\! \Ex\big\{ (\gmkiu)^{\dag} \wimk^{\MRT} \big\} \Big\vert^{2} \Big\} \nonumber\\
    &\hspace{0em}= \!\!\! \sum\nolimits_{m \in \M_\mathtt{I}} \!\! \! \snrdl \etamkI \!\Big(\! \Ex \big\{ \! \big\vert \! (\gmkiu)^{\!\dag} \!\wimk^{\MRT} \big\vert^{2} \! \big\} \!-\!\! \Big\vert \Ex\big\{ (\gmkiu)^{\!\dag} \! \wimk^{\MRT} \big\} \! \Big\vert^{2} \! \Big) 
    \nonumber\\
    &\hspace{0em}= \!\!\! \sum\nolimits_{m \in \M_\mathtt{I}} \! \snrdl \etamkI \Big( \Ex \big\{ \! \big\vert \! (\gmkiu)^{\dag} \! \wimk^{\MRT} \big\vert^{2} \big\} - \alpha^{\MRT}_{m k_i} \Big).
\end{align}
Accordingly, the expectation term in \eqref{eq:BU_k} is derived as 
\begin{align}~\label{eq:mathcalA}
     \Ex \big\{ \big\vert  (\gmkiu)^{\dag}  \wimk^{\MRT} \big\vert^{2} \big\} 
    &=
     \frac{1}{\alpha^{\MRT}_{m k_i}} \Ex\big\{ \Vert \hgmki \Vert^{4}
     \big\}
    \nonumber\\
    &+
    \barbetamki \trace(\FmRIS \FmRIS^{\dag}) - \gamma^{\mathtt{I}}_{m k_i},
\end{align}
where we have exploited Lemma~\ref{lemma1} and the independence property between the estimation error and the precoding vector. Thus, by substituting~\eqref{eq:mathcalA} into~\eqref{eq:BU_k} we obtain $\mathsf{BU}_{k_i}$ in~\eqref{eq:SINE:MRT}.

2) \textit{Compute $\UIki$}: Due to the PC phenomenon, the second term in the denominator of \eqref{eq:SINE:general} is expanded as
\begin{align}
    \!\!\sum_{k' \in \K \setminus k_{i}} & \!\!\!\! \Ex\big\{ \big\vert \UIki \big\vert^2  \big\} 
    \!\!=
    \!\!\!\!\!\!\!\!\!
    \sum_{k'_{i} \in \mathcal{P}_k \setminus k_{i}} \!\!\!\!\!\!\! \Ex\big\{ \big\vert \IUIki \big\vert^2  \big\} 
    \!+
    \!\!\!\!\!\!\!\!
    \sum_{k'_{i} \notin \mathcal{P}_k \setminus k_{i}} \!\!\!\!\!\! \Ex\big\{ \big\vert \IUIki \big\vert^2  \big\} 
    \nonumber\\
    &\hspace{-3em}+
    \!\!
    \sum\nolimits_{k_{e} \in \mathcal{P}_k \setminus k_{i}} \!\!\!\!\!\! \Ex\big\{ \big\vert \EUIki \big\vert^2  \big\} 
    +
    \!\!
    \sum\nolimits_{k_{e} \notin \mathcal{P}_k } \!\!\! \Ex\big\{ \big\vert \EUIki \big\vert^2  \big\}. 
\end{align}

Since the difference lies in the analysis of the expectation terms, we apply a similar process in \eqref{eq:mathcalA} to derive $\Ex \big\{\! \big\vert \big(\gmk \big)^{\dag} \!\qw^{\MRT}_{i, mk'} \big\vert^2 \!\big\}$ for two cases $\forall k,k' \in \K$ as:

i) When $\forall k,k' \notin \mathcal{P}_k, k \neq k'$, we have
    \begin{align}~\label{eq:35}
        \Ex \big\{\! \big\vert \big(\gmk \big)^{\dag} \!\qw^{\MRT}_{i, mk'} \big\vert^2 \!\big\} 
        &\stackrel{(a)}{=} \barbetamk \trace(\FmRIS \FmRIS^{\dag}) 
        \nonumber\\
        &+ \kappa \barbetamk \trace(\FmRIS \FmRIS^{\dag} \zmklos \zmklos^{\dag}) .
    \end{align}

ii) When $\forall k,k' \in \mathcal{P}_k, k\neq k'$, we obtain
    \begin{align}~\label{eq:36}
        \Ex \big\{\! \big\vert \big(\gmk \big)^{\dag} \!\qw^{\MRT}_{i, mk'} \big\vert^2 \!\big\} 
        &\!\stackrel{(b)}{=}\!
        \frac{\bar{b}_{k k'}}{\alpha^{\MRT}_{mk'} } \!+\! \barbetamk \trace(\FmRIS \FmRIS^{\dag} )\! -\! \gamma^{i}_{mk},
    \end{align}
where (a) exploits the independence between the orthogonal channel estimates while Lemma~\ref{lemma2} is used in (b). To this end, by using~\eqref{eq:DS_k},~\eqref{eq:BU_k},~\eqref{eq:35}, and~\eqref{eq:36}, we obtain~\eqref{eq:SINE:MRT}.

\vspace{-0.6em}

\vfill

\end{document}